\documentclass[aps,pra,reprint,superscriptaddress,amsmath,amssymb]{revtex4-1}
\pdfoutput=1

\usepackage{amsthm, graphicx}
\newtheorem{thmcase}{Case}
\newtheorem{thmlemma}{Lemma}

\newcommand{\unit}[1]{\,\mathrm{#1}}

\newcommand{\bra}[1]{\left< #1 \right|}
\newcommand{\ket}[1]{\left| #1 \right>}
\newcommand{\braket}[2]{\langle #1 | #2 \rangle}
\newcommand{\melement}[3]{\left< #1 \right| #2 \left| #3 \right>}
\newcommand{\e}{\mathrm{e}}
\newcommand{\abs}[1]{\lvert #1 \rvert}
\newcommand{\norm}[1]{\lVert #1 \rVert}

\begin{document}

\title{Dynamics of entanglement in a dissipative Bose--Hubbard dimer}

\author{Tadeusz Pudlik}
\affiliation{Department of Physics, Boston University, Boston, MA, 02215, USA}
\author{Holger Hennig}
\affiliation{Department of Physics, Harvard University, Cambridge, MA 02138,
USA}
\author{D. Witthaut}
\affiliation{Network Dynamics, Max Planck Institute for Dynamics and Self-Organization (MPIDS), 37077 G\"ottingen, Germany}
\author{David K.~Campbell}
\email[To whom correspondence should be addressed, at ]{dkcampbe@bu.edu}
\affiliation{Department of Physics, Boston University, Boston, MA, 02215, USA}

\date{\today}

\begin{abstract}
We study the connection between the semiclassical phase space of the Bose--Hubbard dimer and inherently quantum phenomena in this model, such as entanglement and dissipation-induced coherence.  Near the semiclassical self-trapping fixed points, the dynamics of EPR entanglement and condensate fraction consists of beats among just three eigenstates.  Since persistent EPR entangled states arise only in the neighborhood of these fixed points, our analysis explains essentially all of the entanglement dynamics in the system.  We derive accurate analytical approximations by expanding about the strong-coupling limit; surprisingly, their realm of validity is nearly the entire parameter space for which the self-trapping fixed points exist.  Finally, we show significant enhancement of entanglement can be produced by applying localized dissipation.
\end{abstract}
\maketitle

\section{Introduction}
Ultracold gases in optical traps bring unprecedented control and resolution to the study of quantum systems~\cite{Bloch2005, Morsch2006, Lewenstein2007, Bloch2008}.  Their prospective applications range from simulation of solid state phenomena~\cite{Trotzky2008, Bakr2009, Bloch2012} to quantum metrology~\cite{Giovannetti2004} to quantum information processing~\cite{Jaksch1999}, but realizing their potential requires a thorough understanding of quantum coherence.  An opportune system for the study of coherence is a Bose--Einstein condensate loaded into a double-well optical trap, known as a ``BEC dimer'' or ``bosonic Josephson junction'': it is both amenable to theoretical analysis and realizable in current experiments.  Highlights of past work on the BEC dimer include the demonstrations of matter-wave interferometry~\cite{Schumm2005}, number squeezing~\cite{Jo2007, Esteve2008, Gross2010, Julia-Diaz2012, Julia-Diaz2012a} and measurements transcending the standard quantum limit~\cite{Cadoret2009,Lucke2011}, as well as applications such as gravity detectors~\cite{Hall2007} and noise thermometers~\cite{Gati2006a}.

The macroscopic dynamics of the BEC dimer is well described by a semiclassical mean-field model.  The mean-field dynamics, including the emergence of self-trapping fixed points in bifurcations, has been studied both theoretically~\cite{Raghavan1999, Polkovnikov2002, Hennig2010, Hennig2013} and experimentally~\cite{Albiez2005,Zibold2010}.  These works imply a connection between the structure of the classical phase space and truly quantum phenomena such as entanglement, a connection studied in greater detail recently~\cite{Hennig2012}.

Decoherence and dissipation due to interactions with the environment are major obstacles to long-time control of the coherence of quantum systems~\cite{Aliferis2006}.  However, it was recently demonstrated that dissipation can be a versatile tool for the manipulation of quantum systems---if it can be carefully controlled~\cite{Huelga2007, Kraus2008, Verstraete2009}.  In a BEC dimer, interactions with a thermal bath or the escape of atoms from the trap may strengthen rather than destroy coherence~\cite{Trimborn2008, Witthaut2008}.  Work has now begun on analyzing how the effects of noise and phase space structure interact~\cite{Kordas2012}.

In this paper, we use the global phase space picture~\cite{Hennig2012} to offer new insight into the dynamics of entanglement and dissipation-induced coherence in the double-well optical trap.  We show that the wells are entangled in the Einstein-Podolski-Rosen (EPR) sense only in the neighborhood of the classical fixed points.  The time dependence of the entanglement can be entirely explained in terms of beats among three eigenstates of the system.  These eigenstates can only be found numerically, but we derive analytical results in the limit of weak coupling between the wells.  A perturbative expansion in this limit produces excellent agreement with numerically exact calculations.  The dynamics of other observables, such as well population imbalance or condensate fraction, can be understood in the same framework.

These simple patterns in the dimer's behavior are not only interesting in their own right but also suggest a new approach to understanding the less numerically tractable behavior of BECs in multi-well optical lattices.  The self-trapping fixed points of the dimer are analogous to discrete breathers in larger systems.  We hope to pursue this connection in future work.

\section{Entanglement and coherence in the Bose--Hubbard dimer}

Consider a collection of $N$ bosonic atoms in a double-well optical trap sufficiently deep that only the lowest state in each well is populated.  In this so-called two mode approximation, the atoms' dynamics is described by the Bose--Hubbard Hamiltonian~\cite{Milburn1997},
\begin{equation}\label{eq:hamiltonian}
\hat{H} = -J (\hat{a}_1^\dagger \hat{a}_2 + \hat{a}_2^\dagger \hat{a}_1) + \frac{U}{2}\left(\hat{n}_1(\hat{n}_1 - 1) + \hat{n}_2(\hat{n}_2 - 1)\right),
\end{equation}
where $\hat{a}_i$  is the annihilation operator for an atom in well $i$ and $\hat{n}_i = \hat{a}^\dagger_i \hat{a}_i$ is the number operator.  The same Hamiltonian can be realized in related systems, such as two spin states of atoms in a single optical well~\cite{Zibold2010}.  It is also mathematically equivalent to the Lipkin-Meshkov-Glick model~\cite{Lipkin1965, Glick1965, Meshkov1965, Vidal2004}.

The macroscopic dynamics of the BEC dimer is well-described by a mean-field approximation.  This approximation implicitly assumes that the atoms remain at all times in a product state,
\begin{equation}
\ket{z,\, \phi} = \frac{1}{\sqrt{N}} \left(\sqrt{(1+z)/2}\,\hat{a}_1^\dagger + \sqrt{(1-z)/2}\,\e^{\imath\phi}\,\hat{a}_2^\dagger\right)^N \ket{0},
\end{equation}
where $z$ is the population imbalance and $\phi$ is the relative phase of the two modes.  For general quantum states, these observables are defined by $z = \left(\langle \hat{n}_1 \rangle - \langle \hat{n}_2 \rangle \right)/\left(\langle \hat{n}_1 \rangle + \langle \hat{n}_2 \rangle\right)$ and $\langle \hat{a}^\dagger_1 \hat{a}_2 \rangle = \norm{\langle \hat{a}^\dagger_1 \hat{a}_2 \rangle}\,\e^{\imath \phi}$.  In the mean field approximation, the system is described solely in terms of $z$ and $\phi$, and the evolution of these variables is determined by the classical Hamiltonian,
\begin{equation}
H_\mathrm{MF} = \frac{\Lambda z^2}{2} - \sqrt{1-z^2} \cos(\phi),
\end{equation}
where the parameter $\Lambda = U(N-1)/2J$ captures the strength of the repulsive interaction between the bosons \footnote{Some authors~\cite{Raghavan1999, Hennig2012, Polkovnikov2002} choose to use $N$ in the place of $N-1$ in the definition of $\Lambda$.  The two choices are equivalent in the large-$N$ limit, but using $N-1$ gives more accurate predictions of the frequency of small-amplitude motion about the fixed point (Eq.~\ref{eq:f_MF}).}.  The trajectories $(z(t),\, \phi(t))$ are given by the contours of constant ``energy'' $H_\mathrm{MF} = \mathrm{const.}$, shown in Figure~\ref{fig:global_phase_space}(a).  For $\Lambda < 1$, the model has two stable fixed points, at $(z,\,\phi) = (0,\,0)$ and $(0,\,\pi)$.  As $\Lambda$ is increased above $1$, a supercritical pitchfork bifurcation takes place and the stable fixed point at $\phi = \pi$ is replaced with a pair of stable fixed points at $z \neq 0,\,\phi = \pi$ and an unstable fixed point at $(z,\,\phi) = (0,\,\pi)$.
\begin{figure}
\includegraphics[width=\columnwidth]{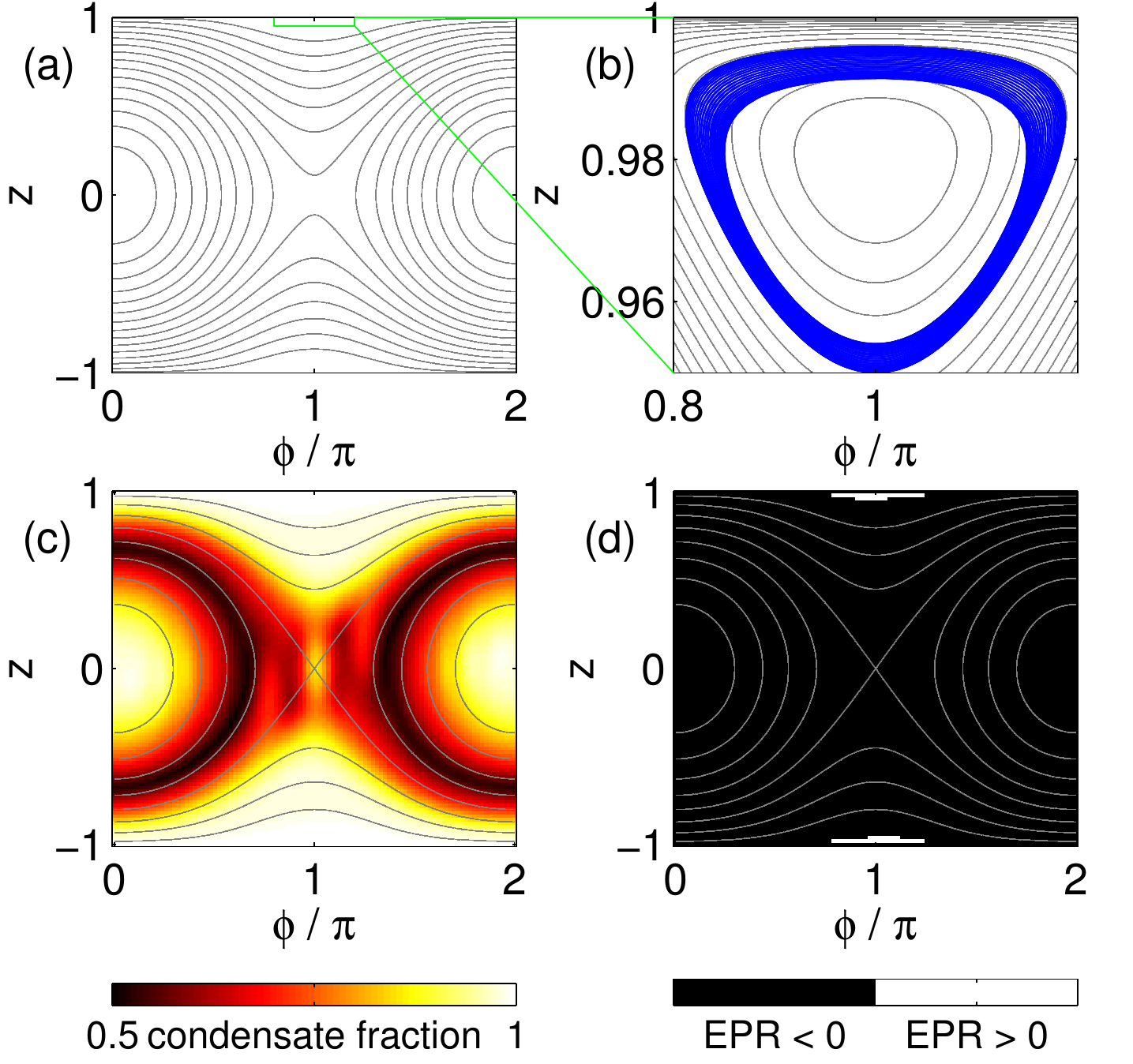}
\caption{(Color online) The global phase space picture of the BEC dimer. \textbf{(a)} Mean-field trajectories. \textbf{(b)} The expectation values of $z$ and $\phi$ over time (phase space trajectory) for an initially coherent state close to the stable fixed point.  The phase trajectory is drawn in blue; nearby mean-field trajectories are in grey.  Note that the actual trajectory has a ``thickness'' associated with it---a phenomenon beyond the mean-field description.\footnote{See Supplemental Material at [URL will be inserted by publisher] for a video of the phase trajectory.} \textbf{(c)} The condensate fraction after 1 second of evolution, for initially coherent states uniformly sampled in $z$ and $\phi$: condensate fraction remains high for initial conditions in the neighborhood of the stable fixed points. \textbf{(d)} EPR entanglement for initially coherent states after 1 second of evolution: the only states still EPR entangled are those initially very near the self-trapping fixed points.  All plots are for $N = 40$, $J = 10\,\hbar/\mathrm{s}$ and $U = 100/39 \approx 2.6\,\hbar/\mathrm{s}$, so $\Lambda = 5$.
\label{fig:global_phase_space}}
\end{figure}

In a global phase space picture we analyze the quantum dynamics of an initially pure BEC $\ket{z,\,\phi}$ as a function of the starting position $(z,\,\phi)$.  The condensate fraction or purity, defined as the largest eigenvalue of the single-particle density matrix,
\begin{equation*}
\rho = \begin{pmatrix}
\langle \hat{a}_1^\dagger \hat{a}_1\rangle & \langle \hat{a}_1^\dagger \hat{a}_2\rangle \\
\langle \hat{a}_2^\dagger \hat{a}_1\rangle & \langle \hat{a}_2^\dagger \hat{a}_2\rangle 
\end{pmatrix},
\end{equation*}
measures how close the many-body state is to a pure BEC~\cite{Witthaut2008,Trimborn2009}.  The EPR entanglement, an important resource in quantum metrology, is quantified by the observable~\cite{Hillery2006, He2011},
\begin{equation}
\mathrm{EPR} = \langle \hat{a}_1^\dagger \hat{a}_2\rangle \langle \hat{a}_2^\dagger \hat{a}_1\rangle - \langle \hat{a}_1^\dagger \hat{a}_1 \hat{a}_2^\dagger \hat{a}_2 \rangle.
\end{equation}
The wells are said to be EPR-entangled whenever $\mathrm{EPR} > 0$.  The global phase space picture~\cite{Hennig2012} shows that the condensate fraction remains large near all of the stable fixed points while EPR entanglement is found only near the $z \neq 0$ fixed points (see Figure~\ref{fig:global_phase_space}).  

The global phase space picture suggests a new method for generating EPR entanglement in the Bose--Hubbard dimer: driving the system closer to the mean-field fixed points using controlled atom loss.  Since the mean-field dynamics is particularly simple near the fixed points, one might hope the full quantum dynamics to be simple as well, allowing for a clear yet quantitative understanding.  To develop such an understanding, in Section~\ref{sec:two_frequencies} we describe the full quantum dynamics of the Bose--Hubbard dimer near the mean-field fixed points, and in Section~\ref{sec:dissipation} consider the effects of controlled atom loss on this dynamics.

\section{Dynamics near the self-trapped fixed points}
\label{sec:two_frequencies}
Let us consider the behavior of the system near the so-called self-trapping fixed points, located at $z = \pm \sqrt{1 - 1/\Lambda^2}$, $\phi = \pi$.  In their neighborhood the observables defined in the previous section exhibit peculiar dynamics, the most striking feature of which is the presence of two distinct frequencies (see Figure~\ref{fig:peculiar_dynamics}).
\begin{figure}
\includegraphics[width=\columnwidth]{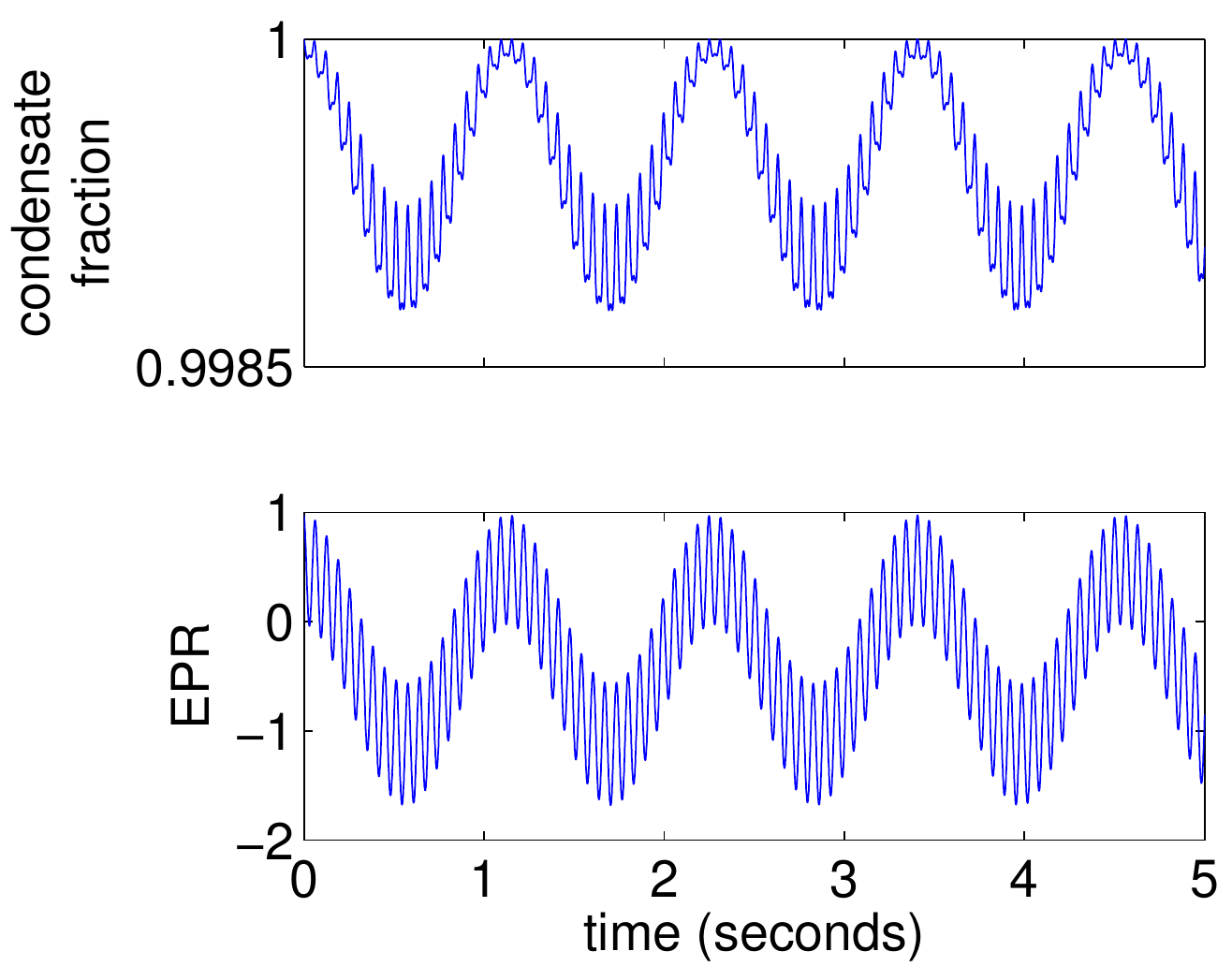}
\caption{(Color online) The condensate fraction and EPR entanglement over time for an initially coherent ($z = 0.95$, $\phi = \pi$) state of the Bose--Hubbard dimer with $N = 40$, $J = 10\,\hbar/\mathrm{s}$ and $U = 100/39 \approx 2.6\,\hbar/\mathrm{s}$, so $\Lambda = 5$.  The stable fixed point is at $z = 2\sqrt{6}/5 \approx 0.98$, $\phi = \pi$.  Results obtained by numerical integration of the Schr\"{o}dinger equation.\label{fig:peculiar_dynamics}}
\end{figure}

The higher frequency is expected on the basis of the mean-field model~\cite{Raghavan1999}.  Linearizing the equations of motion obtained from the Hamiltonian of Eq.~\ref{eq:hamiltonian} about the fixed point yields
\begin{equation}\label{eq:f_MF}
f_\mathrm{MF} = \frac{\sqrt{\Lambda^2 - 1}}{\pi}\frac{J}{\hbar}.
\end{equation}
The mean-field prediction works for a broad range of $\Lambda$ (see Figure~\ref{fig:HF_data}).  The lower frequency, however, cannot be explained within the mean-field approximation.  To see this, consider the trajectory of the system in $z$, $\phi$ space (see Figure~\ref{fig:global_phase_space}(b)).  In the mean-field picture, this trajectory is expected to coincide with the energy contours of $H_\mathrm{MF}$.  However, simulation of the full quantum dynamics reveals a ``thick'' orbit, the size of which oscillates with the low frequency~\footnote{See Appendix~\ref{sec:appendix_quantum_jump} for a discussion of our simulation methods.}.  (Similar low frequency phenomena were noted before~\cite{Milburn1997}, but not discussed quantitatively.)
\begin{figure}
\includegraphics[width=\columnwidth]{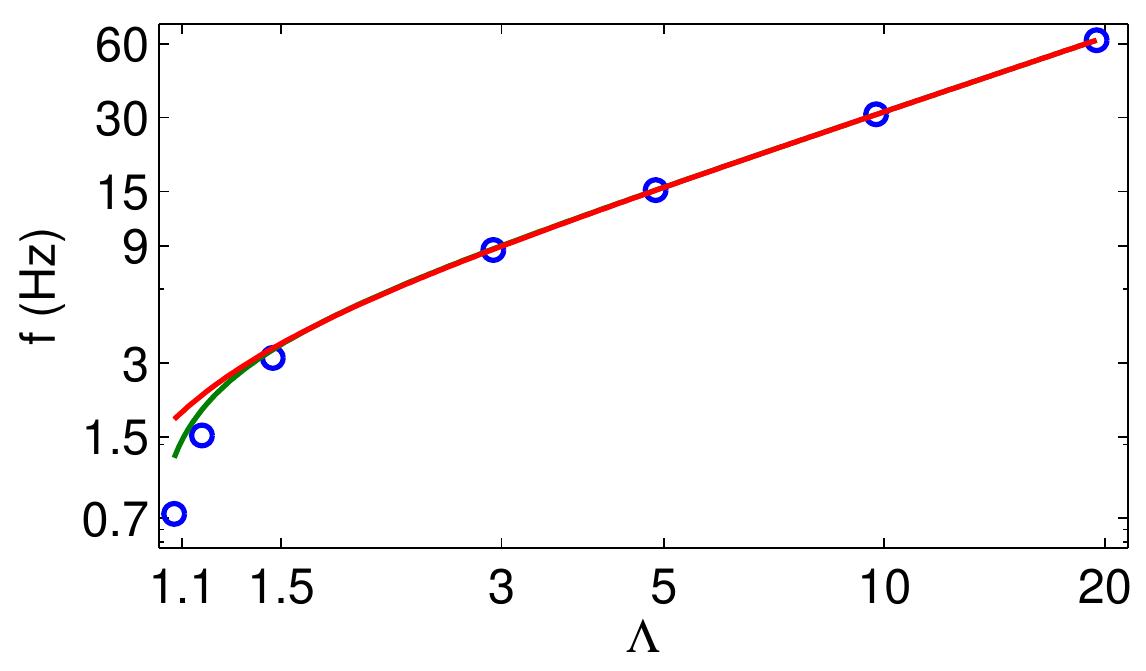}
\caption{(Color online) The high frequency observed near the fixed point for different values of $\Lambda$ (blue circles), on a log-log scale.  The mean-field prediction (green line) is consistent with the numerically exact results.  The $O(\Lambda^{-2})$ perturbative result (show in red) agrees with the mean-field down to $\Lambda = 1.5$, below which it overestimates the frequency; see Section~\ref{sec:J_to_0_limit} for a discussion. \label{fig:HF_data}}
\end{figure}

\subsection{Eigenstate decomposition}
\label{sec:eigenstate_decomposition}
To explain the low frequency oscillations, let us decompose the evolving quantum state into the energy eigenstates:
\begin{equation}\label{eq:eigenstate_decomposition}
\ket{\psi(t)} = \sum_{n=0}^{N} a_n \,\e^{-\imath E_n t/\hbar}\ket{E_n}.
\end{equation}
At first glance, this decomposition does not offer much insight, as the Hamiltonian has a large number of eigenstates and their energies can only be found numerically.  However, in the neighborhood of the system's fixed points only a few states contribute appreciably to the wave function (see Figure~\ref{fig:most_probable_states}).  This is not entirely surprising: the stable fixed points are the extrema of the mean-field energy, so in the neighborhood of these points only the eigenstates with most nearly extremal energy values should contribute to the coherent state. 
\begin{figure}
\begin{center}
\includegraphics[width=\columnwidth]{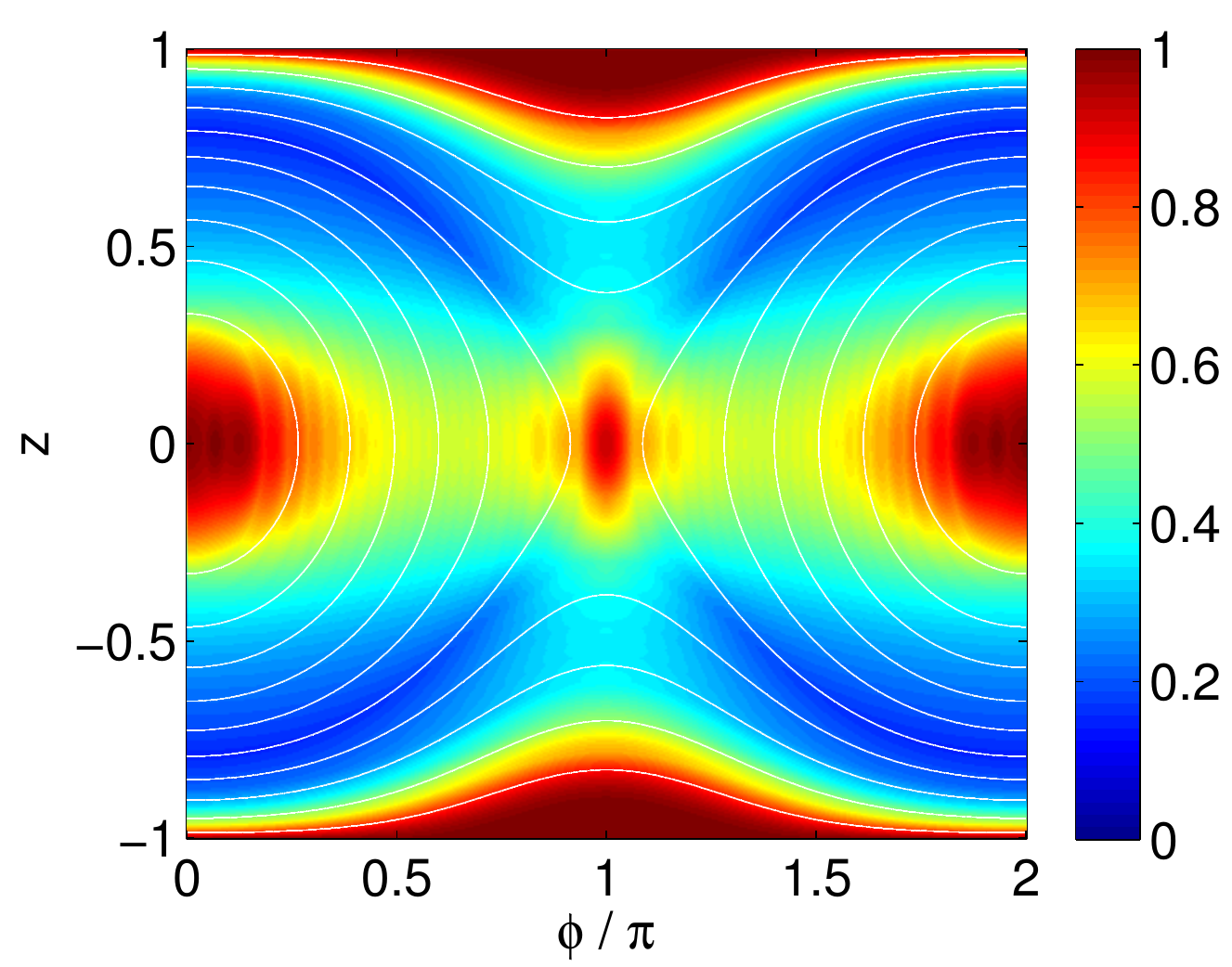}
\caption{(Color online) A projection onto just 3 eigenstates is effective near the mean-field fixed points.  Consider a coherent state $\ket{z,\,\phi} = \sum_{n=0}^{N} a_n \ket{E_n}$ and its projection $\ket{\psi'} = \sum_{n=0}^2 a_n \ket{E_n}$ onto the three energy eigenstates with the largest coefficients $a_n$ in the energy eigenstate expansion.  The plot above shows the norm $\abs{\braket{\psi'}{\psi'}}^2$ of this projection as a function of $z$ and $\phi$.  Note that the norm of the projection is nearly 1 (perfect) near all of the fixed points, including the unstable one, suggesting the three-eigenstate description will be informative for those initial conditions.  The system parameters are the same as those in Figure~\ref{fig:peculiar_dynamics}, namely $J = 10\,\hbar/\mathrm{s}$, $\Lambda = 5$ and $N = 40$.  The mean-field stable fixed points are at $z = 0$, $\phi = 0$ and at $z = \pm 0.98$, $\phi = \pi$, while the unstable fixed point is at $z = 0$, $\phi = \pi$.  Contours of constant mean-field energy are shown in white.  
\label{fig:most_probable_states}}
\end{center}
\end{figure}
 Indeed, for the $z=0.95,\,\phi = \pi$ coherent state of Figure~\ref{fig:peculiar_dynamics}, we find the contributions of the three highest-energy eigenstates to be,
\begin{align*}
a_0 &= 0.9353&	E_0 &= 2040\,\hbar/\mathrm{s}\\
a_1 &= 0.3474&	E_1 &= 1942\,\hbar/\mathrm{s}\\
a_2 &= 0.0653&	E_2 &= 1850\,\hbar/\mathrm{s}.
\end{align*}
These three eigenstates together account for,
\begin{equation*}
\abs{a_0}^2 + \abs{a_1}^2 + \abs{a_2}^2 = 0.9997
\end{equation*}
of the probability weight of the coherent state.  We might therefore expect the frequencies observed in the data to be beats between the eigenstates,
\begin{align*}
(E_0 - E_1)/2\pi &= 15.56\unit{Hz} \equiv f_\mathrm{fast}\\
(E_1 - E_2)/2\pi &= 14.64\unit{Hz}\\
(E_0 - E_2)/2\pi &= 30.23\unit{Hz}
\end{align*}
or perhaps higher-order beats, such as
\begin{equation*}
\frac{E_0-E_1}{2\pi} - \frac{E_1-E_2}{2\pi} = 0.8805\unit{Hz} \equiv f_\mathrm{slow}.
\end{equation*}
These expectations are borne out: $15.6\unit{Hz}$ is the mean-field frequency given by Eq.~\ref{eq:f_MF}, while $0.85(5)\unit{Hz}$ is the measured frequency of the large-amplitude oscillation in Figure~\ref{fig:peculiar_dynamics}.  The two other beats are also seen in the power spectrum of the condensate fraction (at $14.65(5)\unit{Hz}$ and $30.27(5)\unit{Hz}$), though not in that of EPR.~\footnote{See Supplemental Material at [URL will be inserted by publisher] for these power spectra.}

The projection onto the three most important eigenstates recovers not only the frequencies but essentially all of the observables' dynamics (see Figure~\ref{fig:projected_peculiar_dynamics}).  A projection onto just two states is sufficient to recover the mean-field motion, but not the low frequency oscillations.
\begin{figure}
\begin{center}
\includegraphics[width=\columnwidth]{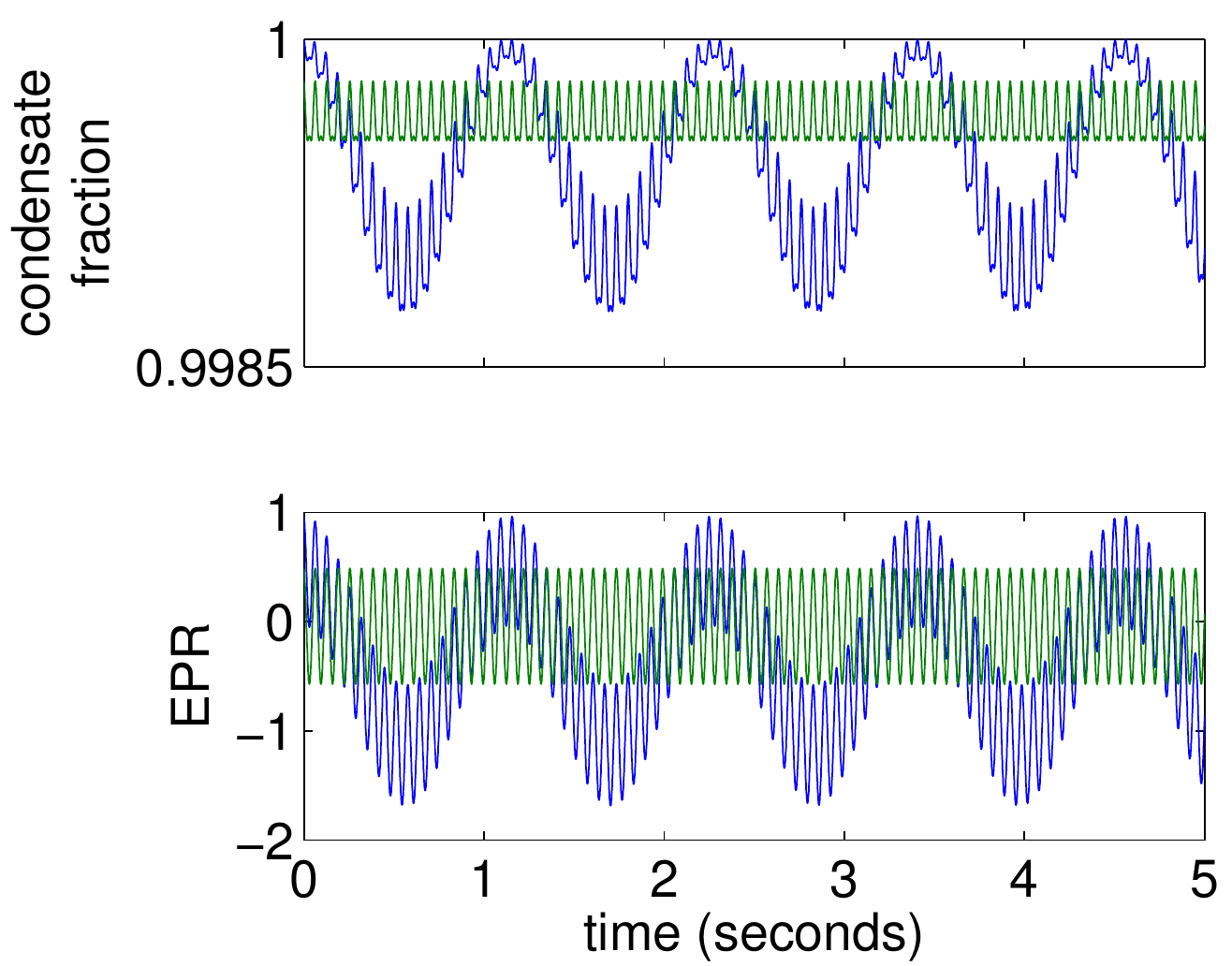}
\caption{(Color online) Validity of the two-frequency approximation.  The condensate fraction and EPR entanglement calculated by approximating the initial coherent state with only 3 eigenstates (blue) is virtually indistinguishable from the numerically exact results shown in Figure~\ref{fig:peculiar_dynamics}.  An approximation with only 2 eigenstates (green) reproduces the fast, but not the slow oscillations.\label{fig:projected_peculiar_dynamics}}
\end{center}
\end{figure}

\subsection{$J\to 0$ limit}
\label{sec:J_to_0_limit}
To gain more insight into the two frequencies, consider the limit $J\to 0$ in which the Hamiltonian can be diagonalized exactly~\footnote{Strictly speaking, the Bose--Hubbard dimer can be analytically solved in the $J\neq 0$ case: a solution based on the Bethe ansatz was developed in the early 1990s~\cite{Enol'skii1991,Enol'skii1992,Links2003}.  This solution replaces the $N+1$ dimensional eigenvalue problem with a set of $N$ nonlinear algebraic equations.  Since for generic $N$ such equations can only be solved numerically, the Bethe ansatz solution amounts to a restatement of our original problem.}.  The eigenstates are the Fock states $\ket{N_1,\,N-N_1} \equiv \ket{N_1}$, and the associated energies are,
\begin{equation}
\epsilon_{N_1} = \frac{N_1(N_1-1)}{2}U + \frac{(N - N_1)(N - N_1 - 1)}{2}U,
\end{equation}
which we denote with $\epsilon$ rather than $E$ to distinguish the $J\to 0$ limit from the general case.  Note the $N_1\to N - N_1$ twofold degeneracy of the spectrum, reflecting the symmetry of the system with respect to a relabeling of the wells.  The frequencies analogous to $f_\mathrm{fast}$ and $f_\mathrm{slow}$ computed numerically in Section~\ref{sec:eigenstate_decomposition} are,
\begin{align*}
\frac{\epsilon_0 - \epsilon_1}{2\pi\hbar} = \frac{U(N-1)}{2\pi\hbar} &\approx 15.9\unit{Hz},\\
\frac{\epsilon_0 - \epsilon_1}{2\pi\hbar} - \frac{\epsilon_1 - \epsilon_2}{2\pi\hbar} = \frac{U}{\pi\hbar} &\approx 0.82\unit{Hz}.
\end{align*}
The $J\to 0$ estimate of $f_\mathrm{fast}$ coincides with the limit of the mean-field expression:
\begin{equation}\label{eq:fast_freq_corrections}\begin{split}
f_\mathrm{MF} &= \frac{\sqrt{\Lambda^2 - 1}}{\pi}\frac{J}{\hbar} = \frac{\sqrt{U^2 (N-1)^2 - 4J^2}}{2\pi\hbar}\\
 &= \frac{U(N-1)}{2\pi\hbar} \sqrt{1-\left(\frac{2J}{U(N-1)}\right)^2}\\
 &= \frac{U(N-1)}{2\pi\hbar}\left[1 - \frac{1}{2}\Lambda^{-2} + O(\Lambda^{-4})\right].
\end{split}\end{equation}
More interesting is $f_\mathrm{slow} = U/\pi\hbar$.  The slow oscillations are a purely quantum phenomenon, as $U/\pi\hbar$ goes to zero in the classical limit of $N\to\infty$ with $\Lambda = U(N-1)/2J$ fixed.  A first hypothesis might identify them with the quantum revivals, in which \emph{all} of the components of the coherent state re-phase~\citep{Greiner2002}.   This is almost correct.  Consider an initial state $\ket{\psi(0)}$ decomposed into energy eigenstates (Eq.~\ref{eq:eigenstate_decomposition}).  Evolving the state over one period $\tau = 1/f_\mathrm{slow}$ of the slow oscillation yields,
\begin{equation}\label{eq:psi_periodicity_maintext}
\ket{\psi(\tau)} = \begin{cases}
\sum_{n=0}^N a_n\ket{n} & \text{for $N$ odd,}\\
\sum_{n=0}^N (-1)^n a_n\ket{n} & \text{for $N$ even,}
\end{cases}
\end{equation}
up to an overall phase (see Appendix for a proof).  For odd $N$, we observe a full revival, as expected.  For $N$ even, the relative phases of the eigenstates are altered, and a revival occurs only after a translation by $2\tau$.  However, the additional phases present after a $\tau$ translation cancel when the condensate fraction and EPR are computed (see Appendix for a proof).  In the limit $J\to 0$ one therefore expects revivals in these observables with a frequency $1/\tau = \frac{U}{\pi\hbar}$ for all values of $N$.

Surprisingly, the $J\to 0$ result is close to the observed frequencies even when $J \gg U$ (see Figure~\ref{fig:LF_data}).  To shed light on this, one may compute the shifts in the frequencies due to $J \gtrsim 0$ using degenerate perturbation theory (see the Appendix for a derivation following~\cite{Bernstein1990}).  The resulting corrections to the $J = 0$ result are proportional to $\Lambda^{-2} \approx (J/NU)^2$:
\begin{equation}
\label{eq:perturbative_corrections}
\begin{split}
\frac{\epsilon_0 - \epsilon_1}{2\pi\hbar} = \frac{U(N-1)}{2\pi\hbar} \bigg[&1 - \frac{1}{2}\frac{N+1}{N-3}\Lambda^{-2} + O(\Lambda^{-4})\bigg],\\
\frac{\epsilon_0 - \epsilon_1}{2\pi\hbar} - \frac{\epsilon_1 - \epsilon_2}{2\pi\hbar} = \frac{U}{\pi\hbar} \bigg[&1 + \frac{3}{2}\frac{(N-1)(N+1)}{(N-5)(N-3)}\Lambda^{-2}\\
& + O(\Lambda^{-4})\bigg].
\end{split}
\end{equation}
The perturbative high frequency estimate agrees with the mean-field result (Eq.~\ref{eq:fast_freq_corrections}) in the limit of large $N$, as one would expect.  Close to the bifurcation the mean-field expression performs better than the perturbative one (see Fig~\ref{fig:HF_data}), presumably because we dropped terms of order $\Lambda^{-4}$ and higher.  But above $\Lambda \approx 2$, the agreement of the perturbative expressions with the observed frequencies of both the mean field motion (Fig~\ref{fig:HF_data}) and the quantum revival (Fig~\ref{fig:LF_data}) is excellent.
\begin{figure}
\includegraphics[width=\columnwidth]{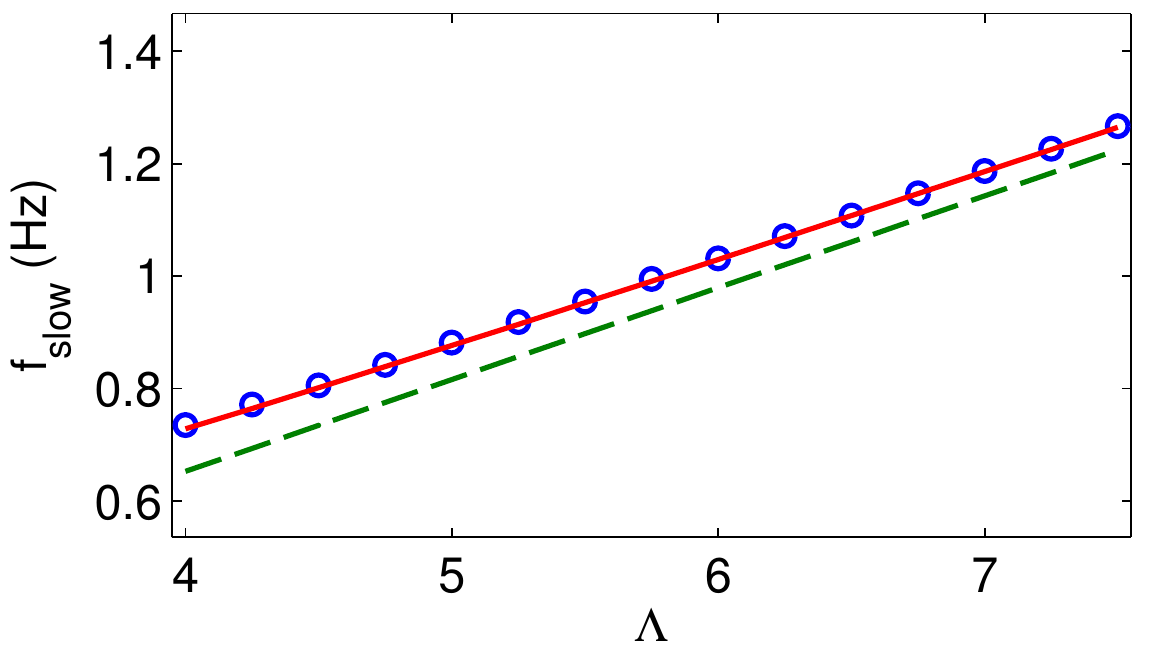}
\caption{(Color online) The slow frequency near the fixed point as a function of $\Lambda$.  The numerically exact values (blue circles) are well described by the second order perturbative results (solid red line).  Zeroth-order perturbation theory (dashed green line) slightly underestimates the frequency.\label{fig:LF_data}}
\end{figure}

\subsection{Region of validity}
How far from the fixed point can we expect the dynamics to be dominated by the two-frequency pattern described above?  To avoid introducing additional frequencies at the outset, the initial coherent state $\ket{z,\,\phi}$ must have an appreciable projection onto just three eigenstates: that is, the projection $\ket{\psi'} = \sum_{n=0}^2 a_n \ket{E_n}$ must satisfy $\norm{\braket{z,\,\phi}{\psi'}}^2 = \norm{\braket{\psi'}{\psi'}}^2 \approx 1$.  But in addition, coherent states at every point of the the mean-field trajectory must be well approximated by the three eigenstates: if $\norm{\braket{z,\,\phi}{\psi'}}^2$ deviates significantly from 1 anywhere along an orbit, a breakdown of the two-frequency pattern is expected.  An instructive example of such a breakdown is observed as the system approaches the bifurcation ($\Lambda \to 1_+$).  Although in the neighborhood of the stable fixed points the norm of the three-eigenstate projection remains high, the mean-field orbits venture out of this neighborhood (see Figure~\ref{fig:most_probable_states_L1}).  The true quantum dynamics involves tunneling from one stable fixed point to the other~\footnote{See Supplemental Material at [URL will be inserted by publisher] for a video showing such tunneling in the Husimi function.} which is classically forbidden and does not conform to the two-frequency paradigm described in this section.
\begin{figure}
\includegraphics[width=\columnwidth]{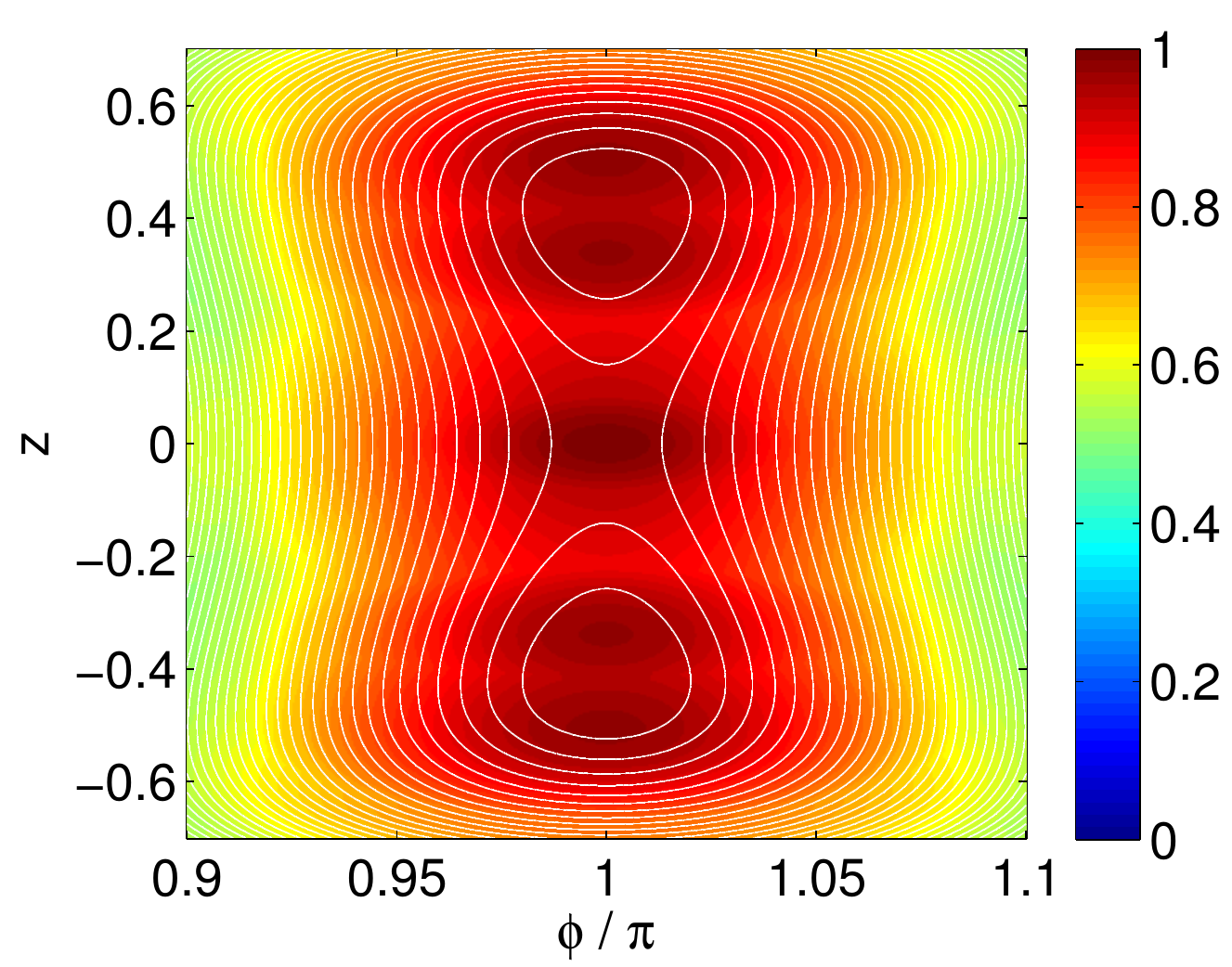}
\caption{(Color online) Breakdown of the three-eigenstate approximation near the bifurcation.  The squared norm $\norm{\braket{\psi'}{\psi'}}^2$ of the projection of coherent states onto the three energy eigenstates with the largest coefficients in the energy eigenstate expansion is plotted, for $\Lambda = 1.1$.  The mean-field trajectories are overlaid in white.  Note that the projection norm is not conserved along the mean-field trajectories, indicating the breakdown of the mean-field approximation and the two-frequency pattern.\label{fig:most_probable_states_L1}}
\end{figure}

Thus, the two-frequency description is valid only for initial conditions in some neighborhood of the stable fixed points.  However, this is generally sufficient to understand the generation of EPR entanglement: On long time scales EPR entanglement is present only for initial conditions close to the fixed points, as shown in Figure~\ref{fig:global_phase_space}.

What is more, the slow oscillations set the timescale for which EPR entanglement is present in the sysem.  To obtain a global picture of entanglement generation, we simulated the dynamics of 10,000 initially coherent states uniformly sampled from the Bloch sphere.  In Figure~\ref{fig:entangled_point_fraction} we plot, as a function of time, the fraction of these in which the two wells are entangled.  Pronounced revivals occur with the frequency $f_\mathrm{slow}$ analyzed above.  The implication, supported by an examination of individual phase space trajectories~\footnote{See Supplemental Material at [URL will be inserted by publisher] for a video of the phase space trajectories.}, is that entanglement is only observed in those regions of phase space where its dynamics is dominated by the two-frequency behavior.  In this sense, the two-frequency model explains the dynamics of the dimer's entanglement quite generically.
\begin{figure}
\includegraphics[width=\columnwidth]{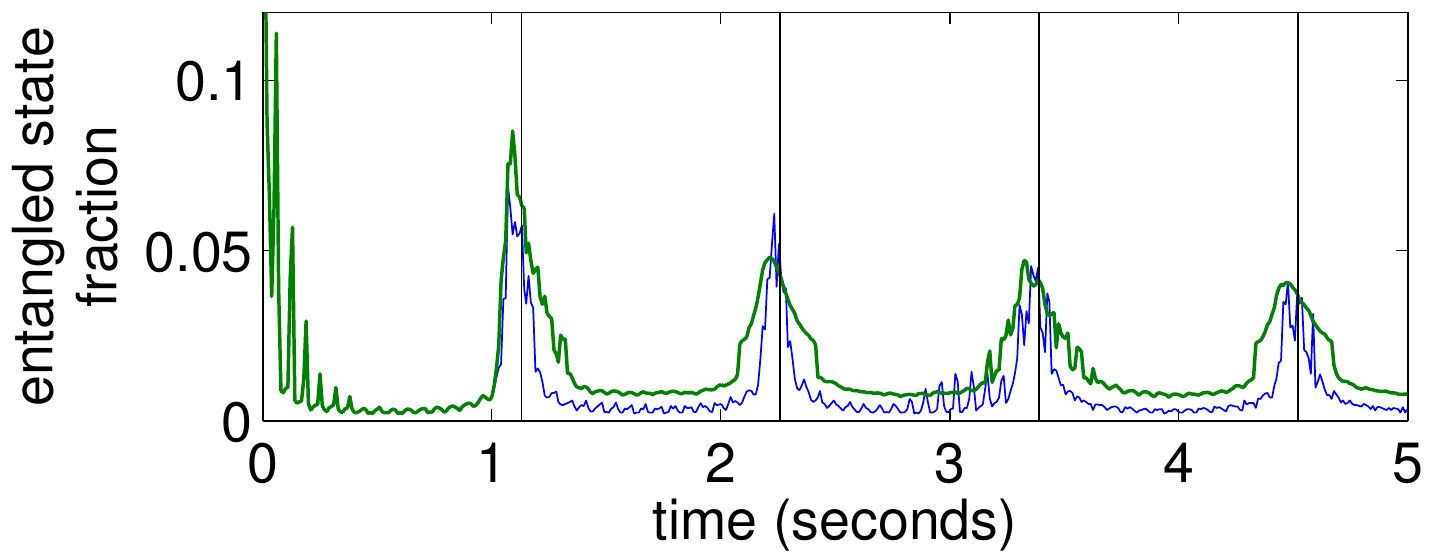}
\caption{(Color online) EPR entanglement is predominatly found at the quantum revival times and can be promoted by applying localized dissipation.  We plot the fraction of coherent state initial conditions for which the two wells are EPR entangled ($\mathrm{EPR} > 0$), as a function of time for $\Lambda = 5$.  The initial conditions are uniformly sampled on the Bloch sphere.  The quantum revival times near the fixed point (multiples of $\tau = 1.13\unit{s}$) are marked with black vertical lines.  The thin blue line is obtained in the absence of dissipation; the thick green line is the result seen when atom loss at the second site is induced between seconds 1 and 1.25 of the simulation.  Applying dissipation increases the fraction of initial conditions for which the wells are persistently entangled. \label{fig:entangled_point_fraction}}
\end{figure} 

\section{Dissipation-induced coherence}
\label{sec:dissipation}
Atoms can be removed from a double-well optical trap with single-site resolution using strong resonant laser blasts or a focused electron beam~\cite{Gericke2008,Wurtz2009}.  This process can be described by the quantum master equation in Lindblad form for the density matrix $\hat{\rho}$~\cite{Breuer2002},
\begin{equation}
\label{eq:master_equation}
\frac{d}{dt}\hat{\rho} = -\imath [\hat{H}, \hat{\rho}] - \frac{1}{2} \sum_{j=1}^2 \gamma_j \left(\hat{a}^\dagger_j \hat{a}_j \hat{\rho} + \hat{\rho} \hat{a}^\dagger_j \hat{a}_j - 2\hat{a}_j \hat{\rho} \hat{a}^\dagger_j\right),
\end{equation}
where $\gamma_j$ is the loss rate at site $j$.  Instead of solving the master equation directly, we use the quantum jump method~\cite{Dalibard1992, Carmichael1993, Molmer1993, Garraway1994, Plenio1998}, discussed further in Appendix~\ref{sec:appendix_quantum_jump}.  Previous studies carried out along these lines show that controlled atom loss may lead to improved coherence (as measured by, among other indicators, the condensate fraction) in the Bose--Hubbard dimer~\cite{Trimborn2008, Witthaut2008} and in multi-well systems (i.e., lattices)~\cite{Trimborn2011, Witthaut2011, Kordas2012}.

An example of dissipation-induced coherence, simulated using the quantum jump method, is shown in Figure~\ref{fig:dissipation_induced_coherence}.  The initial condition is a coherent state near the self-trapping fixed point.  After a second of free evolution, atoms are removed from the less populated site for half a second.  The result is a long-term increase in condensate fraction and a transition from intermittent entanglement to a persistently entangled state.  This process can be understood within the phase space picture: the system's trajectory is driven towards the stable fixed point, which is a region of high entanglement.
\begin{figure*}
\includegraphics[width=\textwidth]{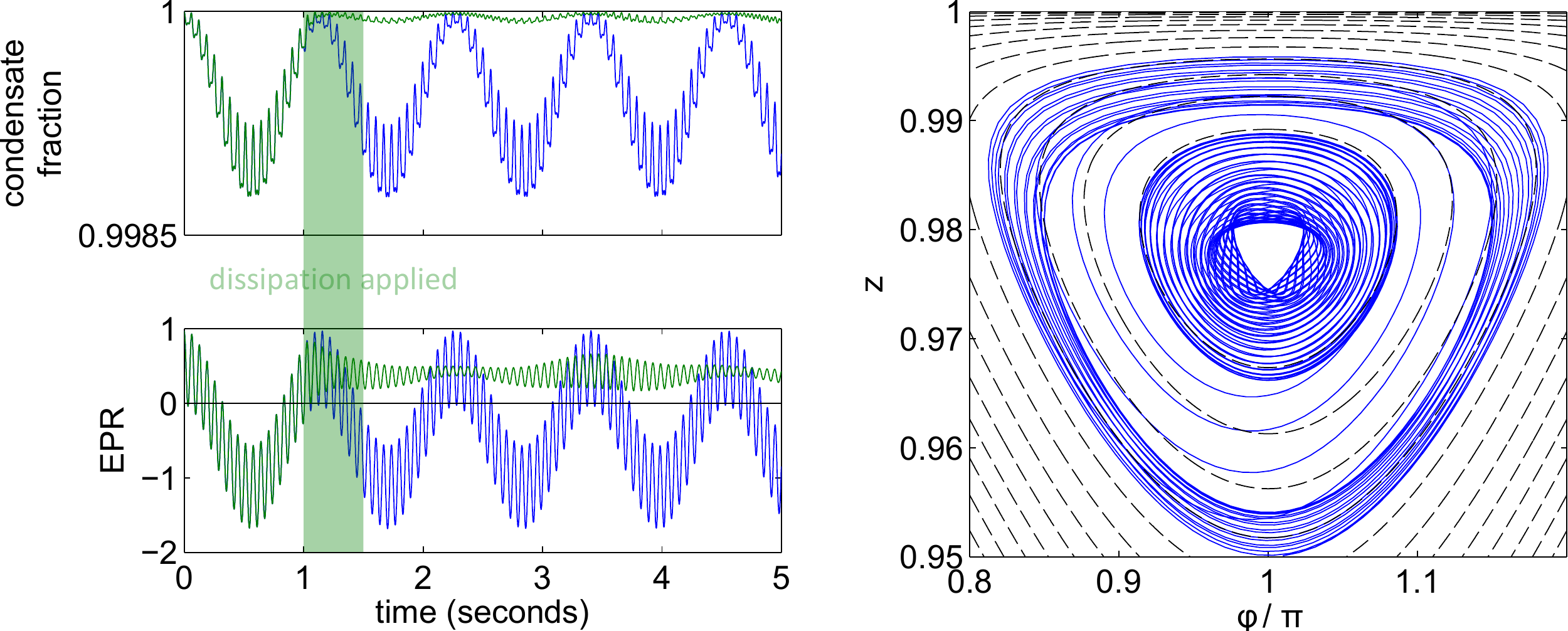}
\caption{(Color online) Dissipation-induced coherence, with signatures in condensate fraction, EPR entanglement and phase space trajectory (expectation values of $z$ and $\phi$).  The initial condition is the same as in Figure~\ref{fig:peculiar_dynamics}, but atom loss ($\gamma_2 = 5\unit{J/\hbar}$) is induced at site 2 between seconds 1 and 1.5 of the simulation.  The plots of condensate fraction and EPR compare the results without (blue) and with (green) dissipation.\label{fig:dissipation_induced_coherence}}
\end{figure*}

How representative is the picture presented above?  Consider again the evolution of 10,000 coherent states uniformly spaced in $z$ and $\phi$, shown in Figure~\ref{fig:entangled_point_fraction}.  Between the quantum revivals, the fraction of entangled states is substantially increased by the application of dissipation.  This implies the mechanism shown in Figure~\ref{fig:dissipation_induced_coherence} operates for an appreciable range of initial conditions.

\section{Summary \& Outlook}

We have used the global phase space picture of the BEC dimer to illuminate the dynamics of entanglement in this system and provide a novel perspective on dissipation-induced coherence.  We showed that for initial conditions close to the mean-field self-trapping points the dimer's dynamics is completely captured by a projection onto just three eigenstates.  Where the projection is successful two frequencies appear prominently in the observables: $f_\mathrm{fast}$, due to the mean-field motion, and $f_\mathrm{slow}$, associated with a quantum revival.  These frequencies are accurately analytically approximated by a second-order expansion about the strong-coupling limit.  The frequency $f_\mathrm{slow}$ sets the dominant time scale for the dynamics of EPR entanglement in the BEC dimer.  This is because the regions of phase space in which our description is valid coincide with the regions where EPR entanglement persists.  It is also within these regions that dissipation-induced entanglement can be induced.

The significance of this work is two-fold.  Firstly, the patterns we describe---two-frequency motion near the fixed point, the driving of the system into the fixed point by dissipation and the resulting enhanced coherence---should be observable in ongoing experiments.  Secondly, and more broadly, analogous patterns may be present in larger, multi-well systems of cold atoms in optical lattices.  These systems are potential platforms for quantum information processing, but by virtue of their size cannot be analyzed via exact techniques.  Consequently, relationships between approximate but tractable semiclassical dynamics and inherently quantum behavior such as are described here offer an attractive path to large-scale quantum engineering.

\nocite{Abdullaev2001,Hennig2010,Freericks1996}

\begin{acknowledgments}
We wish to thank Luca d'Alessio, Pjotrs Gri\v{s}ons and especially Anatoli Polkovnikov for helpful discussions.  This work was supported in part by Boston University, by the U.S.~National Science Foundation under grant No.~PHYS-1066293, and by a grant of the Max Planck Society to the MPRG Network Dynamics.  HH acknowledges support by the German Research Foundation under grant No.~HE 6312/1-1.  We are also grateful for the hospitality of the Aspen Center for Physics.
\end{acknowledgments}

\appendix

\section{Revivals of the wavefunction}

Consider a Bose--Hubbard dimer with $J = 0$ and $N$ atoms.  The energy eigenstates of this system are the Fock states $\ket{n} \equiv \ket{N_1,\,N - N_1}$, with energies
\begin{equation}
\epsilon_n = \frac{n(n-1)}{2}U + \frac{(N-n)(N-n-1)}{2}U.
\end{equation}
It will prove convenient to define a de-dimensionalized energy,
\begin{equation}
h_n \equiv \frac{\epsilon_n}{U}.
\end{equation}
The coherent states of the dimer are of the form,
\begin{equation}
\ket{\psi(t)} = \sum_{n=0}^N a_n \,\e^{-\imath \epsilon_n t/\hbar}\,\ket{n},
\end{equation}
with the expansion coefficients $a_n$ all nonzero.  Consider translating the coherent state in time by
\begin{equation}\label{eq:def:tau}
\tau \equiv \frac{\pi\hbar}{U}.
\end{equation}
As was asserted in the main text, the results of this translation depend on the value of $N$:
\begin{equation}\label{eq:psi_periodicity}
\ket{\psi(\tau)} = \begin{cases}
\ket{\psi(0)} & \text{for $N = 1 + 4p$,}\\
-\sum_{n=0}^N (-1)^n a_n\ket{E_n} & \text{for $N = 2 + 4p$,}\\
-\ket{\psi(0)} & \text{for $N = 3 + 4p$,}\\
\sum_{n=0}^N (-1)^n a_n\ket{E_n} & \text{for $N = 4 + 4p$,}
\end{cases}
\end{equation}
with $p\in\mathbb{Z}_{\geq 0}$.  Let us prove this assertion case by case.

\begin{thmcase}
If $N = 1 + 4p$ for $p\in\mathbb{Z}_{\geq 0}$, then $\ket{\psi(\tau)} = \ket{\psi(0)}$.
\end{thmcase}
\begin{proof}
Note that,
\begin{equation}
\exp\left(-\imath \epsilon_n \tau/\hbar\right) = \exp(-\imath\pi h_n).
\end{equation}
The de-dimensionalized energy $h_n$ satisfies,
\begin{equation*}
\begin{split}
2 h_n &= n(n-1) + (N-n)(N-n-1)\\
 &= n(n-1) + (1 + 4p - n)(1 + 4p - n - 1)\\
 &= n^2 - n + 4p - n +16p^2 - 4np - 4np + n^2\\
 &= 2(n^2 - n) + 4p + 16p^2 - 8pn\\
h_n &= n(n-1) + 2(p + 4p^2 - 2pn)
\end{split}
\end{equation*}
The right-hand side is even for any $n\in \mathbb{Z}_{\geq 0}$ and any $p\in\mathbb{Z}_{\geq 0}$.  Therefore, $h_n$ is even and $\exp\left(-\imath \pi h_n\right) = 1$.  It follows that $\ket{\psi(\tau)} = \ket{\psi(0)}$.
\end{proof}

\begin{thmcase}
If $N = 3 + 4p$ for $p\in\mathbb{Z}_{\geq 0}$, then $\ket{\psi(\tau)} = -\ket{\psi(0)}$.
\end{thmcase}
\begin{proof}
We proceed as in the previous case.
\begin{equation*}
\begin{split}
2 h_n &= n(n-1) + (N-n)(N-n-1)\\
 &= n(n-1) + (3 + 4p - n)(3 + 4p - n - 1)\\
 &= n^2 - n + 6 + 12p - 3n + 8p + 16p^2\\
 &\quad - 4np - 2n - 4np + n^2\\
 &= 2n^2 - 6n + 20p + 16p^2 - 8np + 6\\
h_n &= \left[b(b-1) + 2(-b + 5p + 4p^2 - 2np)\right] + 3
\end{split}
\end{equation*}
Note that the number in square brackets is always even; thus, $h_n$ is odd.  Consequently, for any $n$,
\begin{equation*}
\exp\left(-\imath \pi h_n\right) = -1
\end{equation*}
The claim then follows from the definition of $\ket{\psi(t)}$.
\end{proof}

\begin{thmcase}
Let $N$ be even.  Then,
\begin{equation*}
\begin{split}
\ket{\psi(\tau)} &= \sum_{n=0}^N a_n\,\e^{-\imath \epsilon_n \tau/\hbar}\ket{E_n}\\
 &= \begin{cases}
\sum_{n=0}^N (-1)^n a_n\ket{E_n} & \text{for $N = 4p + 4$,}\\
-\sum_{n=0}^N (-1)^n a_n\ket{E_n} & \text{for $N = 4p + 2$,}
\end{cases}
\end{split}
\end{equation*}
with $p\in\mathbb{Z}_{\geq 0}$.
\end{thmcase}
\begin{proof}
Let $q = N/2$.  We have,
\begin{equation*}
\begin{split}
2 h_n &= n(n-1) + (2q-n)(2q - n - 1)\\
 &= 2n^2 + 4q^2 - 4nq - 2q\\
h_n &= (n^2 - q) + 2(q^2 - nq) = (n^2 - q) + \text{even factor}
\end{split}
\end{equation*}
The factor $n^2 - q$ (and, by extension, $h_n$) is even if and only if $n$ and $q$ are of the same parity.  Recall that $q$ is even when $N = 4p + 4$ and odd when $N = 4p + 2$.  Therefore,
\begin{equation*}
\begin{split}
\exp\left(-\frac{\imath \epsilon_n \tau}{\hbar}\right) &= \exp(-\imath \pi h_n) \\
&= \begin{cases}
1  & \parbox[t]{0.25\textwidth}{if $N = 4p + 4$ and $n$ is even or $N = 4p + 2$ and $n$ is odd,}\\
-1 & \parbox[t]{0.25\textwidth}{if $N = 4p + 4$ and $n$ is odd or $N = 4p + 2$ and $n$ is even.}
\end{cases}
\end{split}
\end{equation*}
The claim follows immediately from the definition of $\ket{\psi(t)}$.
\end{proof}

\section{Revivals of the condensate fraction and EPR}

Consider the two observables discussed in the main text.  The condensate fraction is given by the normalized largest eigenvalue of the single-particle density matrix,
\begin{equation}\label{eq:app:c}
c = \frac{1}{2N}\left(\rho_{11} + \rho_{22} + \sqrt{(\rho_{11} - \rho_{22})^2 + 4\rho_{12}\rho_{21}}\right),
\end{equation}
where $\rho_{ij} = \langle \hat{a}_i^\dagger \hat{a}_j \rangle$.  The entanglement measure EPR is,
\begin{equation}\label{eq:app:EPR}
\mathrm{EPR} = \langle \hat{a}_1^\dagger \hat{a}_2\rangle \langle \hat{a}_2^\dagger \hat{a}_1\rangle - \langle \hat{a}_1^\dagger \hat{a}_1 \hat{a}_2^\dagger \hat{a}_2 \rangle.
\end{equation}
In this section, we show that these observables take the same values at time $t = \tau = \pi\hbar/U$ as at $t = 0$, regardless of the value of $N$.  For $N$ odd, this follows immediately from the results of the first section of this Appendix, so assume $N$ even.

Consider first the condensate fraction.  We have,
\begin{equation*}
\rho_{ij}(t = 0) = \sum_n \sum_m a_n^* a_m \melement{E_n}{\hat{a}_i^\dagger \hat{a}_j}{E_m}
\end{equation*}
and, by the result proven in the first section,
\begin{equation*}
\rho_{ij}(t = \tau) = \sum_n \sum_m (-1)^{n+m} a_n^* a_m \melement{E_n}{\hat{a}_i^\dagger \hat{a}_j}{E_m}.
\end{equation*}
At $J = 0$ the energy eigenstates are the Fock states, so the components of $\rho$ can be found immediately:
\begin{align*}
\rho_{11}(t=\tau) &= \sum_n \sum_m (-1)^{n+m} a_n^* a_m n \delta_{n,\,m}\\
 &= \sum_n (-1)^{2n} \abs{a_n}^2 n = \sum_n \abs{a_n}^2 n = \rho_{11}(t=0)\\
\rho_{22}(t=\tau) &= \rho_{22}(t = 0)\quad\text{(analogously)}\\
\rho_{12}(t=\tau) &= \sum_{n,m} (-1)^{n+m} a_n^* a_m \times \\
&\qquad\quad \sqrt{(m+1)(N-m)}\delta_{n,\,m+1}\\
 &= \sum_n (-1)^{2n-1} a_n^* a_{n-1} \sqrt{n(N-n+1)}\\
 &= -\sum_n a_n^* a_{n-1} \sqrt{n(N-n+1)}\\
 &= -\rho_{12}(t=0)\\
\rho_{21}(t=\tau) &= -\rho_{21}(t=0)\quad\text{(Hermiticity)}
\end{align*}
Note that $\rho(t=\tau)$ differs from $\rho(t=0)$ only in the sign of the off-diagonal elements.  But these elements enter the condensate fraction only through their product (cf.~Eq.~\ref{eq:app:c}).  Therefore, $c(t = \tau) = c(t = 0)$.

The argument for EPR is similar.  The first term in Eq.~\ref{eq:app:EPR} is equal to $\rho_{12}\rho_{21}$, and so the same at $t = \tau$ as at $t = 0$; the second term is equal to $N_1 N_2$, the product of the wells' populations, and so independent of time.  Therefore, $\mathrm{EPR}(t = \tau) = \mathrm{EPR}(t = 0)$.

\section{Perturbation about the $J\to 0$ limit}

In this Appendix, we derive the perturbative corrections to the $J = 0$ mean-field and quantum revival frequencies (Equation~\ref{eq:perturbative_corrections}).

It is convenient to rescale the problem by dividing all energies by $NU$.  Then, the unperturbed Hamiltonian is given by,
\begin{equation}
H_0 = \frac{1}{2N}\left(\hat{a}_1^\dagger \hat{a}_1^\dagger \hat{a}_1 \hat{a}_1 + \hat{a}_2^\dagger \hat{a}_2^\dagger \hat{a}_2 \hat{a}_2 \right)
\end{equation}
and is its representation in the Fock basis is,
\begin{equation}
\begin{pmatrix}
\frac{N(N-1)}{2N} & & & & \\
& \frac{(N-1)(N-2)}{2N} & & & \\
& & \frac{(N-2)(N-3)}{2N} + \frac{1}{N} & & \\
& & & \ddots & \\
& & & & \frac{N(N-1)}{2N}
\end{pmatrix}
\end{equation}
In analogy to the main text, we will denote the diagonal entries $\varepsilon_0$, $\varepsilon_1$, and so on; because of the rescaling of the Hamiltonian, $\epsilon_i = NU \varepsilon_i$. The perturbed Hamiltonian is
\begin{equation}
H = H_0 + \frac{J}{NU}\left(-\hat{a}_1^\dagger \hat{a}_2 - \hat{a}_2^\dagger \hat{a}_1 \right) \equiv H_0 + \lambda V
\end{equation}
and the Fock basis representation of $V$ is the tridiagonal matrix,
\begin{equation}
-\begin{pmatrix}
0 & \sqrt{N} & & & \\
\sqrt{N} & 0 & \sqrt{2(N-1)} & & \\
& \sqrt{2(N-1)} & 0 & \sqrt{3(N-2)} & \\
& & & \ddots & \\
& & & \sqrt{N} & 0
\end{pmatrix}
\end{equation}
Let $P$ be a projection operator onto a subspace corresponding to a set of degenerate levels of $H_0$.  Number the levels $n = 0,\,1,\,2,\,\ldots$, as in the main text; then, this operator is represented by a matrix with only two nonzero entries, $P_{n+1,n+1} = P_{N+1-n,N+1-n} = 1$.  Degenerate perturbation theory can be used to show that the $k$th order corrections to the level energies, $\epsilon^{(k)}$, are the eigenvalues of the matrix~\cite{Bernstein1990},
\begin{equation}
P W_k P
\end{equation}
where the first few $W_k$ matrices are
\begin{equation}
\label{eq:Ws}
\begin{split}
W_1 &= V, \\
W_2 &= -V L^{-1} V,\\
W_3 &= V(L^{-1} V)^2,\\
W_4 &= -V(L^{-1} V)^3 - \epsilon^{(2)} V(L^{-1})^2 V,
\end{split}
\end{equation}
with $L^{-1}$ represented by a diagonal matrix with entries,
\begin{equation}
(L^{-1})_{ll} = \begin{cases}
1/(H_0 - \epsilon^{(0)} I)_{ll} & \text{if $(H_0 - \epsilon^{(0)} I)_{ll} \neq 0$,}\\
0 & \text{otherwise,}
\end{cases}
\end{equation}
where $\epsilon^{(0)}$ are the unperturbed energies.

Since $V$ has no diagonal entries, $P W_1 P$ is a matrix of zeroes and there are no first-order corrections.  (That $\epsilon^{(1)} = 0$ is used in the expressions for $W_3$ and $W_4$ in Eq.~\ref{eq:Ws}, which would otherwise contain terms proportional to $\epsilon^{(1)}$.)  The second order corrections can be computed using the matrix $W_2$.  Its nonzero entries are,
\begin{widetext}
\begin{equation}
(W_2)_{ij} = \begin{cases}
-N/(\varepsilon_1-\varepsilon_n) & \text{for $i = j = N+1$ or $i = j = 1$,} \\
-\frac{(i-1)(N+2-i)}{\varepsilon_{i-1} - \varepsilon_n} - \frac{i(N+1-i)}{\varepsilon_i - \varepsilon_n} & \text{for other $i = j$,}\\
-\sqrt{i(i+1)(N+1-i)(N-i)}/(\varepsilon_{i+1} - \varepsilon_n) & \text{for $j = i + 2$,}\\
-\sqrt{j(j+1)(N+1-j)(N-j)}/(\varepsilon_{j+1} - \varepsilon_n) & \text{for $j = i - 2$.}
\end{cases}
\end{equation}
\end{widetext}
where $n$ is the index of the level considered.

Recall that the degenerate states are those corresponding to rows $i$ and $N+1-i$, for $i = 1,\,2,\,\ldots$ .  Since $W_2$ has nonzero entries only on the main diagonal and the $\pm 2$ diagonals, the projection onto the subspace of degenerate levels $P W_2 P$ may have off-diagonal entries only for the second-lowest energy level.  Conversely, if $N \gg 1$, the second-order corrections to the higher energy levels are given by the corresponding diagonal entries of $W_2$.  The corrections to the three highest-energy levels up to second order are given by,
\begin{equation}
\begin{split}
\varepsilon_0 &\to \varepsilon_0 - \frac{N}{\varepsilon_1 - \varepsilon_0}(J/NU)^2 \\
     &\quad \varepsilon_0 + \frac{N^2}{N-1} (J/NU)^2 \\
\varepsilon_1 &\to \varepsilon_1 - \left(\frac{N}{\varepsilon_0 - \varepsilon_1} + \frac{2(N-1)}{\varepsilon_2 - \varepsilon_1}\right) (J/NU)^2 \\
     &\quad \varepsilon_1 + \frac{N(N^2-N+1)}{(N-3)(N-1)} (J/NU)^2\\
\varepsilon_2 &\to \varepsilon_2 - \left(\frac{2(N-1)}{\varepsilon_1 - \varepsilon_2} + \frac{3(N-2)}{\varepsilon_3 - \varepsilon_2}\right)(J/NU)^2\\
 &\quad \varepsilon_2 + \frac{N(N^2 - 3N + 8)}{(N-5)(N-3)}(J/NU)^2
\end{split}
\end{equation}
These expressions lead immediately to Eq.~\ref{eq:perturbative_corrections}.

Consider now the third-order corrections to the energies, determined by the matrix $W_3$.  Except for the lowest-lying states, there are no off-diagonal entries in $P W_3 P$; this follows from the central result of~\cite{Bernstein1990}, since off-diagonal entries would break the degeneracy the authors prove to hold to $N$th order of perturbation theory.  Therefore, the third-order corrections to the energies are given by the diagonal entries of the matrix $W_3$.  As we show below, these entries are all zero, and thus there are no corrections of this order.

Let $A$ be a matrix. We will call $A$ an \emph{odd matrix} if $A_{i,i+p} = 0$ for all $i$ and all $p$ even (including $p = 0$), and an \emph{even matrix} if $A_{i,i+p} = 0$ for all $i$ and all $p$ odd.  Even and odd matrices have the following properties under multiplication:
\begin{thmlemma}
Let $A$ and $B$ be odd $N\times N$ matrices.  Then, $AB$ is an even matrix.
\end{thmlemma}
\begin{proof} Let $n$ be an odd integer.
\begin{equation*}
(AB)_{i,i+n} = \sum_{k=1}^N A_{i,k} B_{k,i+n} = \sum_{p=1-i}^{N-i} A_{i,i+p} B_{i+p,i+n}
\end{equation*}
Since $A$ is an odd matrix, we may restrict the sum to odd $p$ (the other entries are zero):
\begin{equation*}
(AB)_{i,i+n} = \sum_{p\,\mathrm{odd}} A_{i,i+p} B_{i+p,i+n}
\end{equation*}
Let $m = i+p$.
\begin{equation*}
B_{i+p,i+n} = B_{m,m+n-p}
\end{equation*}
Since $p$ is odd and $n$ odd, $n-p$ is even and so $B_{i+p,i+n} = 0$ because $B$ is an odd matrix.  Thus,
\begin{equation*}
(AB)_{i,i+n} = 0
\end{equation*}
for $n$ odd, and $AB$ is an even matrix.
\end{proof}
\begin{thmlemma}
Let $A$ be an odd $N\times N$ matrix and $C$ an even $N\times N$ matrix.  Then, $AC$ and $CA$ are odd matrices.
\end{thmlemma}
\begin{proof}
Let $n$ be an even integer.  Proceeding as before,
\begin{equation*}
\begin{split}
(AC)_{i,i+n} &= \sum_{k=1}^N A_{i,k} C_{k,i+n} = \sum_{p=1-i}^{N-i} A_{i,i+p} C_{i+p,i+n}\\
 &= \sum_{p\,\mathrm{odd}} A_{i,i+p} C_{i+p,i+n}.
\end{split}
\end{equation*}
Let $m = i+p$.
\begin{equation*}
C_{i+p,i+n} = C_{m,m-p+n}.
\end{equation*}
Since $p$ is odd and $n$ is even, $n-p$ is odd.  Since $C$ is an even matrix, $C_{i+p,i+n} = 0$ and
\begin{equation*}
(AC)_{i,i+n} = 0
\end{equation*} 
for $n$ even, proving that $AC$ is an odd matrix.  The proof for $CA$ is analogous.
\end{proof}

The matrix $V$ is odd, while the matrix $L^{-1}$ is even.  By the lemmas above, $W_3 = V(L^{-1}V)(L^{-1}V)$ must be odd.  The diagonal entries of an odd matrix are all zero.  This implies there are no third-order corrections to the energies, except for the lowest few energy levels.  Incidentally, the same argument can be used to show that the fifth-order corrections are zero.

\section{The quantum jump method}
\label{sec:appendix_quantum_jump}

This Appendix outlines the quantum jump method, the technique used in our simulations of the BEC dimer in the presence of dissipation.

We start with a simulation time interval $[0,\, T]$ and an initial state described by a wavefunction $\ket{\psi(t=0)}$.  The time interval is divided into time steps $\delta t$.  At each time step we act on the state with the time evolution operator,
\begin{equation*}
U = \exp\left(-\frac{\imath}{\hbar} \hat{H'} \delta t\right),
\end{equation*}
where the pseudo-Hamiltonian $\hat{H'}$ consists of the standard Bose-Hubbard Hamiltonian $\hat{H}$ of Equation~\ref{eq:hamiltonian} and an extra term:
\begin{equation*}
\hat{H'} = \hat{H} - \frac{\imath\hbar}{2}\sum_{j=1,2} \gamma_j \hat{a}^\dagger_j \hat{a}_j.
\end{equation*}
In addition to the time evolution, each time step atoms are removed from the wells $j = 1,\, 2$ with probability
\begin{equation*}
\delta p_j = \delta t \cdot \gamma_j \cdot \bra{\psi(t)} \hat{a}_j^\dagger \hat{a}_j \ket{\psi(t)},
\end{equation*}
where $\gamma_j$ the atom loss rate from well $j$.  The time $\delta t$ is taken to be sufficiently small that the possibility of multiple atoms being removed during one time step can be ignored.  If an atom is removed, we update the wavefunction,
\begin{equation*}
\ket{\psi} \to \hat{a}_j \ket{\psi},
\end{equation*}
and reduce the dimension of the Hilbert space.  (The Hilbert space
of the dimer has dimension $N + 1$, where $N$ is the number of particles, so the problem shrinks as atoms are ejected.)  Finally, regardless of whether an atom was removed or not, we renormalize the wavefunction.

A surprising feature of this algorithm is that the time evolution differs from the dissipation-free case even when no atoms are removed.  This is because the absence of a removal event reveals information about the system, changing the probability distribution over the well occupation numbers and so altering the wave functions.  Intuitively, if strong dissipation is applied to a well and yet no atoms are ejected from it, the well is likely to be empty.  Alternately, one may interpret the suppression of tunneling into the well as a manifestation of the quantum Zeno effect~\cite{Cirac1994}.  See Section 3 of~\cite{Molmer1993} for a more detailed discussion and references devoted specifically to this paradoxical ``null measurement'' effect.

\section{A note on projections}

The phase space of the semiclassical model of the BEC dimer is a sphere; representing it in a plane, as we have done in Figures~\ref{fig:global_phase_space},~\ref{fig:most_probable_states}, ~\ref{fig:most_probable_states_L1} and~\ref{fig:dissipation_induced_coherence}, requires a choice of projection.  Throughout this work, we have used the cylindrical equal-area projection, also known as the Lambert projection.  In the Lambert projection, the $x$ and $y$ coordinates on the map are proportional to the longitude $\phi$ (polar angle) and the cosine of the colatitude $z = \cos \theta$ (the usual $z$ coordinate of spherical coordinates).  In other words, the Lambert is an axial projection of the sphere onto a cylinder tangent at the equator.

As its name suggests, the cylindrical equal-area projection preserves areas: any two regions on the sphere with equal areas $S_1 = S_2$ are mapped to two regions on the map with equal areas $S'_1 = S'_2$.  However, this projection deforms angles (it is not conformal), as a circle on the sphere is mapped to an ellipse, the eccentricity of which increases towards the poles.

The Lambert projection should not be confused with the Mercator projection, defined by
\begin{align*}
x &= \phi\\
y &= \ln \tan \left(\frac{\pi - \theta}{2}\right)
\end{align*}
The Mercator projection is conformal, mapping a circle anywhere on the sphere to a circle in the $x,\,y$ plane, but is not area-preserving: famously, Mercator maps of the Earth show South America to be smaller than Greenland, though the former land mass is in fact some eight times larger.  In addition, since $y\to\infty$ as $\theta\to 0$ or $\pi$, a Mercator map must employ a $\theta$ cutoff and cannot include the poles.  Some authors~\cite{Trimborn2009, Chuchem2010} have described their maps of the Bloch sphere as being drawn in the Mercator projection, but labeled the $y$ scale from $z = -1$ to $z = 1$, suggesting they might have actually used the less well-known Lambert (in which $y\propto z$ and the poles are in the range of the map).

A thorough and freely available reference on projections of the sphere is the review publication of the US Geological Survey,~\cite{Snyder1987}.

%

\end{document}